\documentclass[%
 reprint,
 amsmath,amssymb,
 aps,
]{revtex4-2}
\newtheorem{theorem}{\bf Theorem}[section]

\newtheorem{definition}{\bf Definition}[section]

\newtheorem{proof}{\bf Proof}[section]
\usepackage{amssymb}
\def\qed{\hfill $\Box$}
\usepackage[caption=false]{subfig}
\usepackage{graphicx}
\usepackage{dcolumn}
\usepackage{bm}
\usepackage{hyperref}
\usepackage[]{enumitem}
\hypersetup{
    colorlinks,%
    citecolor=blue,%
    filecolor=blue,%
    linkcolor=blue,%
    urlcolor=blue
}

\begin{document}

\preprint{APS/123-QED}

\title{Conflict-free joint sampling for preference satisfaction through quantum interference}

\author{Hiroaki Shinkawa$^1$}
\email{gokukyukyoku555@g.ecc.u-tokyo.ac.jp}
\author{Nicolas Chauvet$^1$}
\author{Andr\'{e} R\"{o}hm$^1$}
\author{Takatomo Mihana$^1$}
\author{Ryoichi Horisaki$^1$}
\author{Guillaume Bachelier$^2$}
\author{Makoto Naruse$^1$}
\affiliation{%
$^1$ Graduate School of Information Science and Technology, The University of Tokyo,\\
  7-3-1 Hongo, Bunkyo-ku, Tokyo 113-8656, Japan\\
$^2$ Univ. Grenoble Alpes, CNRS, Institut N\'{e}el, 38000 Grenoble, France
}%


\date{\today}

\begin{abstract}
Collective decision-making is vital for recent information and communications technologies. In our previous research, we mathematically derived conflict-free joint decision-making that optimally satisfies players' probabilistic preference profiles. However, two problems exist regarding the optimal joint decision-making method. First, as the number of choices increases, the computational cost of calculating the optimal joint selection probability matrix explodes. Second, to derive the optimal joint selection probability matrix, all players must disclose their probabilistic preferences. Now, it is noteworthy that explicit calculation of the joint probability distribution is not necessarily needed; what is necessary for collective decisions is sampling. This study examines several sampling methods that converge to heuristic joint selection probability matrices that satisfy players' preferences. We show that they can significantly reduce the above problems of computational cost and confidentiality. We analyze the probability distribution each of the sampling methods converges to, as well as the computational cost required and the confidentiality secured.
In particular, we introduce two conflict-free joint sampling methods through quantum interference of photons.
The first system allows the players to hide their choices while satisfying the players' preferences almost perfectly when they have the same preferences.
The second system, where the physical nature of light replaces the expensive computational cost, also conceals their choices under the assumption that they have a trusted third party.\\
This paper has been published in Ref. \cite{PhysRevApplied.18.064018} (DOI: 10.1103/PhysRevApplied.18.064018).
\end{abstract}

\maketitle

\section{Introduction}
The problem of allocating indivisible commodities, that is, resources that cannot be divided into multiple parts, such as houses and people, has been studied for a long time \cite{shapley1974cores, svensson1999strategy, alkan1991fair, bogomolnaia2001new}.
One of the most well-known studies is the Top Trading Cycle (TTC) proposed by Shapley and Scarf \cite{shapley1974cores}. 
TTC is a deterministic allocation method that deals with situations in which players have deterministic preference rankings over their options.
From a game-theoretic perspective, TTC achieves what is called core, namely, a situation in which exchanging options among arbitrary players does not lead to a more preference-satisfying allocation.

In the previous study, we  extended the preference from a deterministic to a probabilistic one and mathematically discussed how a joint selection that satisfies players' probabilistic preferences should be made \cite{shinkawa2022optimal}.
In uncertain situations, people in the real world or agents in reinforcement learning are often torn between the desire to choose the best current option and the desire to explore other options \cite{march1991exploration}.
In such situations, they will not be satisfied with obtaining only the top preference option all the time. Instead, they will be satisfied if the proportion of options obtained through multiple allocations matches their probabilistic preferences.

Specifically, let $p_{i, j}$ be the joint probability of assigning the $i$-th choice to the first player A and the $j$-th choice to the second player B, and let the matrix of these joint probabilities be called the joint selection probability matrix.
By definition, the probability of the first player choosing each option $i$ as a result of the joint selection probability matrix can be calculated by
\begin{equation}
  \pi_A(i) = \sum_{j=1}^N p_{i,j},
\end{equation}
where $N$ is the number of options. For the allocation to satisfy player A's preference, the list of $\pi_A(i)$ should coincide with his/her probabilistic preference.

In Ref. \cite{shinkawa2022optimal}, a joint selection probability matrix that maximizes the satisfaction of the probabilistic preferences of two players is mathematically derived.
However, as we will confirm later, two concerns exist. First, the computational cost of obtaining the optimal joint selection probability matrix is $\mathcal{O}(N^2)$ in the worst case if we follow the algorithm presented in the paper, which implies that computing the matrix becomes more difficult as the number of options becomes huge.
The second problem is that of confidentiality. Since the construction of the optimal joint selection probability matrix requires information on the preferences of both players, each player must disclose their preference to the other or third parties.

Now, when we encounter situations of collective decisions without choice conflict, it is not always necessary to explicitly calculate the values of the joint selection probability matrix. 
Instead, employing a sampling method that converges to those values over many repetitions is often sufficient.
Although an efficient sampling method that is based on the optimal joint selection probability matrix has not been established yet, this paper proposes several sampling methods each of which converges to a heuristic joint selection probability matrix over many repetitions.
In particular, we demonstrate that the application of quantum systems can significantly reduce the two problems mentioned above, that is, the explosion of computational cost and the lack of confidentiality.

Optical computing, which flourished around the 1980s \cite{wetzstein2020inference}, smoldered somewhat due to the rapid advances in electronic technology, but is now drawing attention again due to the increasing demand for computational resources caused by AI and so on in recent years \cite{kitayama2019novel, shastri2021photonics, feldmann2021parallel}. It is now being considered for a wide range of applications, including deep learning and computational science, taking advantage of not only the high-speed and broadband nature of light but also its quantum nature \cite{lin2018all, van2017advances, o2007optical}.
For example, a lot of combinatorial optimization problems are regarded as NP-hard, and thus it is difficult for a digital computer to solve them as the size of the problem increases.
However, for some types of problems, even huge combinatorial optimization problems can be solved within short time by mapping them to corresponding Ising models \cite{barahona1982computational, lucas2014ising} and solving the models using Ising machines, for example, with networks of optical parametric oscillators \cite{mcmahon2016fully, bohm2019poor}. 
Optically implemented Ising machines have their advantages over other types of Ising machines, including their ability to work at room temperature \cite{marandi2014network}, high efficiency due to light's broad bandwidth \cite{inagaki2016coherent}, and so on.
While Ising models are often used to solve combinatorial problems such as the Max-Cut problem and the traveling salesman problem, other photonic implementations can be considered for different types of problems. Among them, recent attempts have been made to utilize the quantum nature of light to solve a decision-making problem called the multi-armed bandit problem \cite{sutton2018reinforcement}.

The multi-armed bandit problem is one of the simplest reinforcement learning problems, and it is a question of how decisions should be made in uncertain situations.
Specifically, given multiple slot machines, each with its own probability of generating a reward, the question is how to maximize the cumulative rewards by drawing one of these machines at each time step.
Since the player does not know the hit probabilities a priori, one of the efficient algorithms is to make decisions based on a probabilistic preference so that he/she can both exploit the current best option and explore other options.

The problem is further complicated when multiple players participate in the bandit problem \cite{besson2018multi, lai2010cognitive, kim2016harnessing}. 
In the competitive bandit problem, when multiple players draw the same machine and that machine generates a reward, the reward is split and distributed among them. 
In such a situation, if we consider the expected value of the total reward, we can see that it is always better if the players' choices do not overlap.

What typically happens in the competitive bandit problem is that if each player draws a machine according to only his/her own selection probability, selection conflicts will occur frequently and the final cumulative rewards will be reduced.
However, the quantum nature of light can be used to link individual reward maximization with total reward maximization. Chauvet {\it{et al.}} devised a system that uses entanglement of polarization to prevent selection conflicts without direct communication between the players, and experimentally demonstrated the effectiveness of this system to tackle the competitive multi-armed bandit problem \cite{chauvet2019entangled, chauvet2020entangled}.

The advantage of utilizing the quantum nature of light here is twofold. The first is that it enables the players to conduct probabilistic decision-making through the observation of polarized light. Specifically, the stochasticity associated with the observation of polarization can be linked to probabilistic decision-making by mapping the choice of the machine to the polarization observed in a way that if the photon is detected by an avalanche photodiode corresponding to the horizontally polarized light, the player will select the first machine and vice versa \cite{naruse2015single}. The second advantage of using the quantum nature is that entanglement guarantees that the players' choices never overlap. 
The team reward will not be diminished thanks to the non-conflict decisions by the players \cite{chauvet2019entangled}. Furthermore, Amakasu {\it{et al.}} theoretically showed that conflict-free collective decision-making is possible over an arbitrary number of choices by employing orbital angular momentum of light to overcome the limitations of the number of choices in the case of polarization-based approaches \cite{amakasu2021conflict}. Orbital angular momentum is another degree-of-freedom associated with photons. It carries theoretically infinite numbers of states, and is widely utilized in applications such as optical communications \cite{allen2016optical, willner2015optical}.

Those previous studies aimed to maximize total cumulative rewards by making conflict-free decisions. 
Conversely, the present research, as well as the related former work \cite{shinkawa2022optimal}, exclude external factors such as rewards. 
Instead, the focus is on how to accomplish the maximization of preference satisfaction; that is, how well the player's preference is reflected in the joint decision.

In this study, we demonstrate that quantum systems can be utilized in the preference satisfaction problem.
In Sec. \ref{subsec:prob_settings}, we first review the problem settings of probabilistic preference satisfaction. 
In the subsequent Sec. \ref{subsec:optimal_thm}, we review theorems about the optimal joint selection matrix and clarify the issues related to the construction of the optimal joint selection matrix.
After that, we propose and demonstrate sampling methods that converge to heuristic joint selection probability matrices.
In particular, Sec. \ref{sec:heuristics} covers two sampling methods through quantum interference, each of which is analyzed in detail in terms of implementation, computational cost, confidentiality, and the joint selection probability matrix it converges to.
In Sec. \ref{sec:demo}, we compare the losses (defined in Sec. \ref{sec:prev}) of the joint selection probability matrices to which the sampling methods converge through numerical calculations.
Finally, Sec. \ref{sec:conclusion} provides a summary of this research and future perspectives.

\section{Preference satisfaction by conflict-free joint decisions}\label{sec:prev}
\subsection{Problem settings}\label{subsec:prob_settings}
In this section, we review the problem settings of the conflict-free probabilistic preference satisfaction proposed in Ref. \cite{shinkawa2022optimal}.
Suppose that two players, player A and B, have probabilistic preferences over $N$ options ($N \geq 2$). Let $A_i$ be player A's preference for option $i$ and $B_i$ be player B's preference for option $i$. Since $A_i$ and $B_i$ are probabilities, the following constraints are satisfied:
\begin{gather}
  A_{1}+A_{2}+\cdots+A_{N}=B_{1}+B_{2}+\cdots+B_{N}=1,\\
  A_{i} \geq 0, \quad B_{i} \geq 0 \quad(i=1,2, \cdots, N).
\end{gather}

Let $p_{i,j}$ be the probability of player A choosing option $i$ and player B choosing option $j$ as a result of collective decision-making. $p_{i,j}$ must satisfy the following conditions:
\begin{equation}
  \sum_{i, j} p_{i, j}=1, \quad p_{i, j} \geq 0.
\end{equation}
Then, we define the joint selection probability matrix $\boldsymbol{P}$ such that the element $(i,j)$ of $\boldsymbol{P}$ is $p_{i,j}$:
\begin{equation}
  \boldsymbol{P}=\left(\begin{array}{cccc}
    0 & p_{1,2} & \cdots & p_{1, N} \\
    p_{2,1} & 0 & \cdots & \vdots \\
    \vdots & \vdots & \ddots & \vdots \\
    p_{N, 1} & \cdots & \cdots & 0
    \end{array}\right).
\end{equation}
The diagonal elements are all zero because we deal with collective decision-making without choice conflict.

The property we demand for $\boldsymbol{P}$ is the following. 
The probability that player A can choose option $i$ as a result of $\boldsymbol{P}$ is obtained by summing over columns $j$:
\begin{equation}
  \pi_{A}(i)=\sum_{j} p_{i, j}.
\end{equation}
We call $\pi_{A}(i)$ the satisfied preference, and if this value is consistent with player A's original preference $A_i$, it means that the preference is satisfied for option $i$.
Similarly, we can obtain the satisfied preference for player B by summing over rows $i$:
\begin{equation}
  \pi_{B}(j)=\sum_{i} p_{i, j}.
\end{equation}
Our goal is to determine $p_{i,j}$ that make the satisfied preference as close as possible to the original preferences of both players for all options.
In other words, our objective is to find $p_{i,j}$ so that they realize
\begin{equation}
  \pi_{A}(i) \approx A_{i}, \quad \pi_{B}(j) \approx B_{j}
\end{equation}
for all $i = 1, 2, \ldots , N$ and $j = 1, 2, \ldots , N$.

Figure \ref{fig:problem_settings} schematically illustrates the problem settings. 
Players A and B have probabilistic preferences over options when $N=4$, where in this case,
\begin{eqnarray}
    A_1 = 0.1, \quad A_2 = 0.2, \quad A_3 = 0.3, \quad A_4 = 0.4, \\
    B_1 = 0.3, \quad B_2 = 0.2, \quad B_3 = 0.2, \quad B_4 = 0.3.
\end{eqnarray}
Then, an algorithm calculates a joint selection probability matrix $\boldsymbol{P}$. A decent algorithm should output a matrix in such a way that the satisfied preferences match the players' preferences. 
For example, if we take the sum of the second row of the joint selection probability matrix in the red shaded area, we should get a value close to player A's preference towards the second option.
Similarly, if we take the sum of the third column in the green shaded area, we should get a value close to player B's preference towards the third option.
An example of a joint selection probability matrix that satisfies the players' preferences perfectly is
\begin{equation}
    \boldsymbol{P} = \begin{pmatrix}
        0&0&0&0.1\\0&0&0.1&0.1\\0&0.2&0&0.1\\0.3&0&0.1&0
    \end{pmatrix}.
\end{equation}

\begin{figure*}[htp]
  \includegraphics[width=12.0cm]{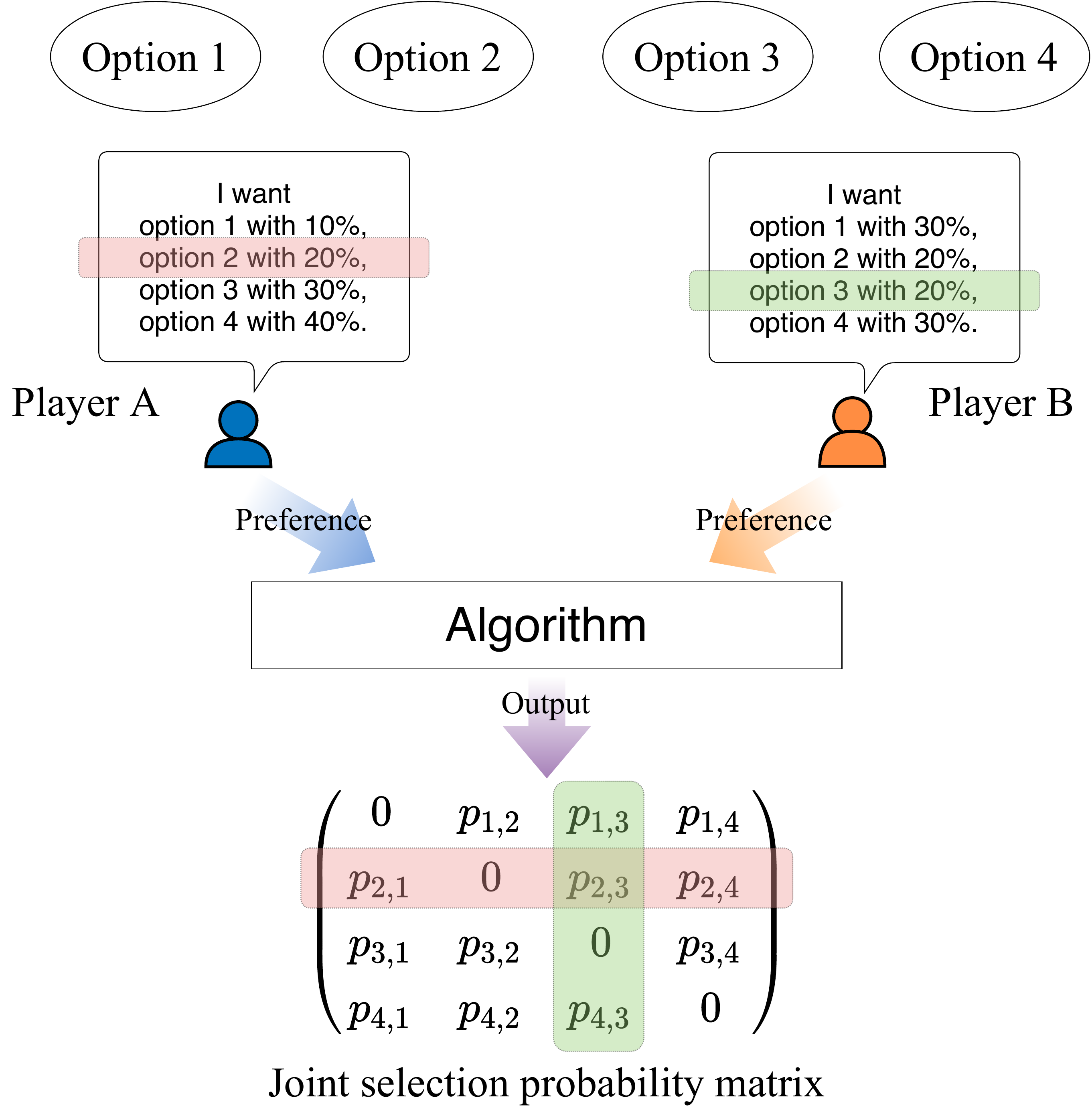}
  \caption{Problem settings. Two players have probabilistic preferences and an algorithm calculates a joint selection probability matrix that satisfies their preferences. The sum of each row should be close to the corresponding preference of player A, and the sum of each column should be close to the corresponding preference of player B.}
  \label{fig:problem_settings}
\end{figure*}

Finally, to quantify the degree of preference satisfaction, we define the degree of deviation between the satisfied preferences and the original preferences as the loss.
The loss $L$ is defined in a manner analogous to the $L_2$-norm as follows:
\begin{equation}
  L=\sum_{i}\left(\pi_{A}(i)-A_{i}\right)^{2}+\sum_{j}\left(\pi_{B}(j)-B_{j}\right)^{2},
  \label{eq:L2_norm_loss}
\end{equation}
which is composed of the sum of squares of the gap between the satisfied preference and the original preference.
The smaller the loss is, the more successfully the preferences are satisfied, and when the loss is zero, the player's preferences are perfectly satisfied.

\subsection{Optimal joint selection probability matrix}\label{subsec:optimal_thm}
In this section, we review principal theorems on the optimal joint selection probability matrix from \cite{shinkawa2022optimal} and clarify two problems associated with them.
We define a score called the popularity, which represents how much each option is favored by the players.
\begin{definition}
  \normalfont The popularity $S_i$ is defined as the sum of the preferences of player A and player B for option $i$.
    \begin{equation}
        S_i := A_i + B_i \quad (i=1, 2, \ldots, N).
    \end{equation}
    Since the preferences $A_i$ and $B_i$ are probabilities, it holds that
    \begin{equation}
        \sum_i S_i = \sum_i A_i + \sum_i B_i = 2.
    \end{equation}
\end{definition}
Two theorems have been found regarding the popularity $S_i$.
\begin{theorem}\label{thm:0loss_case}
  Assume that all the popularities $S_i$ are smaller than or equal to 1. Then, it is possible to construct a joint selection probability matrix that makes the loss $L$ equal to zero.
  \begin{equation}
      \forall i; S_i \leq 1 \Rightarrow L_{\normalfont\text{min}} = 0.
  \end{equation}
\end{theorem}
\begin{theorem}\label{thm:non0loss_case}
  If any value of $S_i$ is greater than 1, it is not possible to make the loss $L$ equal to zero.\\
  In a case when the $N$th option is the most popular, that is, $\max{\{S_i\}}=S_N>1$, the minimum loss is
  \begin{equation}\label{eq:minloss_greater_than_1}
      L_{\text{min}} = \frac{N}{2(N-1)}\cdot (S_N-1)^2.
  \end{equation}
  The following joint selection probability matrix is one of the matrices that minimize the loss.
  \begin{gather}
      \tilde{\boldsymbol{P}} =
          \begin{pmatrix}
              0&0&\cdots&0&A_1+\epsilon\\
              0&0&\cdots&0&A_2+\epsilon\\
              \vdots&\vdots&\ddots&\vdots&\vdots\\
              0&0&\cdots&0&A_{N-1}+\epsilon\\
              B_1+\epsilon&B_2+\epsilon&\cdots&B_{N-1}+\epsilon&0
          \end{pmatrix},\\
          \epsilon=\frac{S_N-1}{2(N-1)}.
  \end{gather}
\end{theorem}
The existence of the optimal joint selection probability matrices and their specific construction methods were presented for both cases where the maximum popularity is less than or equal to one, and where it is greater than one.

However, there are two concerns that need to be resolved.
The first is the computational cost of constructing the optimal joint selection probability matrix in Theorem \ref{thm:0loss_case}. To fill in one row and column of the matrix, we need to:
\begin{enumerate}
  \item Determine the maximum and minimum values of the popularities.
  \item Fill in at most $N-1$ elements.
\end{enumerate}
Determining the maximum and minimum values requires a computational cost of $\mathcal{O}(N)$ each, and filling in at most $N-1$ elements requires a computational cost of $\mathcal{O}(N)$ because every time each element is filled we need one subtraction.
Therefore, it requires $\mathcal{O}(N^2)$ to fill all the rows and columns in the joint selection probability matrix.
Thus, when $N$ becomes huge, it is difficult to compute the optimal joint selection probability matrix.

The second concern, common to the construction of the optimal joint selection probability matrix in both Theorems \ref{thm:0loss_case} and \ref{thm:non0loss_case}, is the lack of confidentiality.
In both cases, constructing the optimal joint selection probability matrix requires the players' preferences $A_i$ and $B_i$. 
This means that the players must disclose their preferences to each other or a third party.
In the real world, the necessity of preference disclosure is undesirable if the players do not know each other or have no means of communication.
For example, few people feel like telling their preferences over sensitive matter to someone they do not know. Even if they do, they may not have a way to connect.

Here, we consider how to deal with these two problems. 
Practically, if we consider conflict-free collective decision-making, we do not necessarily need to calculate the values of the joint selection probability matrix explicitly. Instead, it will be enough if there is a sampling method that converges to that matrix over repeated draws.
In this study, we demonstrate two quantum sampling methods that converge to heuristic joint selection probability matrices with relatively small losses and analyze how each of them deals with the above problems. In particular, we show how quantum interference effects can realize conflict-free joint sampling while highly satisfying individual preference profiles and resolving the confidentiality issue. It should be emphasized that the physical processes, not computers, play the role of establishing conflict-free joint sampling while taking account of individual preferences. The computing cost is replaced by the physical nature of light.
We have not yet devised a sampling algorithm that always converges to the optimal joint selection probability matrix. 
However, as demonstrated later in Sec. \ref{sec:demo}, one of the proposed quantum samplings realizes almost comparable performances to the optimal cases under certain conditions.  

\section{Joint sampling methods through quantum interference}\label{sec:heuristics}
In this section, we propose the following two sampling methods, each of which converges to a joint selection probability matrix with relatively small loss $L$:
\begin{enumerate}[label=\Alph*.]
    \item Pure Hong-Ou-Mandel (Pure HOM)
    \item Orbital Angular Momentum Attenuation \\(OAM Attenuation)
\end{enumerate}
Both of them employ orbital angular momentum (OAM) of light, which is a degree-of-freedom that consists of a theoretically infinite number of states, and they are relatively easy to implement using basic equipment such as spatial light modulators and beam splitters \cite{amakasu2021conflict}.

As introduced in Sec. \ref{sec:prev}, there were two problems in constructing the optimal joint selection probability matrix: high computational cost and low confidentiality.
For each of the above sample methods, we analyze the following four features:
\begin{enumerate}
  \item Implementation
  \item Computational cost
  \item Confidentiality
  \item Joint selection probability matrix that it converges to
\end{enumerate}

\subsection{Pure Hong-Ou-Mandel}
\subsubsection{Implementation}
The method we call ``Pure Hong-Ou-Mandel (Pure HOM)'' employs a system based on the Hong-Ou-Mandel effect involving quantum interference of orbital angular momentum of photons \cite{hong1987measurement, bouchard2020two}.
Figure \ref{fig:mod_hom} schematically illustrates the quantum system we use to make conflict-free joint decisions.

\begin{figure}[htp]
  \includegraphics[width=7.0cm]{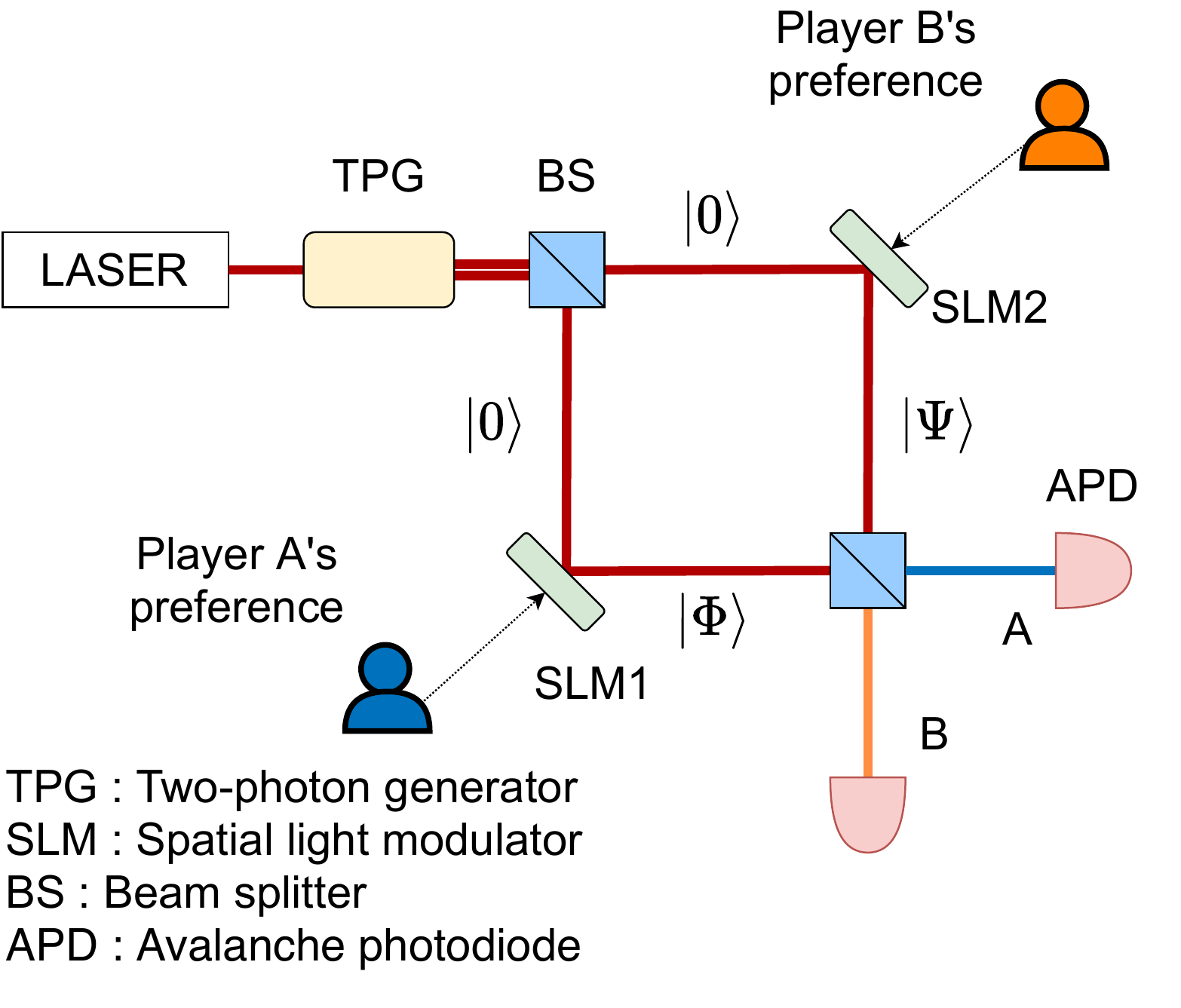}
  \caption{Implementation of Pure HOM. Each player adjusts one of the SLMs using only their own preference profile. Photodetectors are placed immediately after the output of the HOM effect. The detected OAM number is mapped to the index of the option.}
  \label{fig:mod_hom}
\end{figure}

First, a photon pair is created by a two-photon generator, and split into two paths by a beam splitter. At this point, each photon does not carry orbital angular momentum. Orbital angular momentum states can be induced by a spatial light modulator (SLM), which displays computer-generated holograms on its surface \cite{yao2011orbital, wang2012terabit}. 

In the proposed system, player A controls SLM1 to encode his/her probabilistic preference to a photon. Since, SLMs can manipulate both amplitude and phase terms of OAM states, the resulting photon can be described by
\begin{equation}\label{eq:phom_phi}
    |\Phi\rangle=\sum_{k=1}^{K} a_k e^{i \phi_{k}}|+k\rangle, \quad \sum_k a_k^2 = 1.
\end{equation}
Similarly, player B controls SLM2 to encode his/her preference to the other photon:
\begin{equation}\label{eq:phom_psi}
    |\Psi \rangle=\sum_{k=1}^{K} b_k e^{i \psi_{k}}|-k\rangle, \quad \sum_k b_k^2 = 1.
\end{equation}

Then, the photon pair is simultaneously injected into a beam splitter, where the Hong-Ou-Mandel effect happens. 
The output OAM states are the tensor product of the OAM states of the two injected photons:
\begin{align}\label{eq:mod_hom_output}
  \begin{split}
    |\Phi, \Psi\rangle= \sum_{k=1}^{K} \frac{i}{2} &a_{k} b_{k} e^{i\left(\phi_{k}+\psi_{k}\right)}|+k\rangle_{A A}^{2} \\
    +\sum_{k_{1}<k_{2}} \frac{i}{2}&\left(a_{k_{1}} b_{k_{2}} e^{i\left(\phi_{k_{1}}+\psi_{k_{2}}\right)}\right.\\
    &\left.+a_{k_{2}} b_{k_{1}} e^{i\left(\psi_{k_{1}}+\phi_{k_{2}}\right)}\right) \left|+k_{1},+k_{2}\right\rangle_{A A}\\
    +\frac{1}{2} \sum_{k_{1}=1}^{K} \sum_{k_{2}=1}^{K}&\left(a_{k_{1}} b_{k_{2}} e^{i\left(\phi_{k_{1}}+\psi_{k_{2}}\right)}\right.\\
    &\left.-a_{k_{2}} b_{k_{1}} e^{i\left(\psi_{k_{1}}+\phi_{k_{2}}\right)}\right)\left|+k_{1},-k_{2}\right\rangle_{A B} \\
    +\sum_{k=1}^{K} \frac{i}{2} &a_{k} b_{k} e^{i\left(\phi_{k}+\psi_{k}\right)}|-k\rangle_{B B}^{2} \\
    +\sum_{k_{1}<k_{2}} \frac{i}{2}&\left(a_{k_{1}} b_{k_{2}} e^{i\left(\phi_{k_{1}}+\psi_{k_{2}}\right)}\right.\\
    &\left.+a_{k_{2}} b_{k_{1}} e^{i\left(\psi_{k_{1}}+\phi_{k_{2}}\right)}\right) \left|-k_{1},-k_{2}\right\rangle_{B B}.
  \end{split}
\end{align}

As shown in Fig. \ref{fig:mod_hom}, photodetectors are placed immediately after the beam splitter.
The probability of detecting $|+i\rangle$ on side A and $|-j\rangle$ on side B is
\begin{align}\label{eq:phom_rate}
  \begin{split}
  r_{i,j} &= \left|\frac{1}{2}\left(a_{i} b_{j} e^{i\left(\phi_{i}+\psi_{j}\right)}-a_{j} b_{i} e^{i\left(\psi_{i}+\phi_{j}\right)}\right)\right|^{2}\\
  &=\frac{1}{4}\left(a_{i}^{2} b_{j}^{2}+a_{j}^{2} b_{i}^{2}-2 a_{i} a_{j} b_{i} b_{j} \cos \left(\theta_{i}-\theta_{j}\right)\right),
  \end{split}
\end{align}
\begin{equation}
    \theta_k = \frac{\phi_k-\psi_k}{2}.
\end{equation}
Substituting $i=j$ verifies that the same absolute values of the OAM number will never be observed on sides A and B, whatever parameters $a_k, b_k, \phi_k, \psi_k$ the players use for the input states.

Now, in order to realize joint sampling, the observed OAM number is mapped to the index of choice.
For example, if the OAM of $+1$ is observed on side A, player A will select the first option, and if $|-2\rangle$ is observed on side B, player B will select the second option.
Selection conflicts between the two players will never happen under this rule, because the absolute values of the observed OAM number on sides A and B always differ thanks to the Hong-Ou-Mandel effect.
Therefore, with Pure HOM realized by Fig. \ref{fig:mod_hom}, collective decision-making without selection conflicts can be achieved.

Note that there are cases where two photons come out on the same side. For example, according to Eq. \eqref{eq:mod_hom_output}, the probability of both photons, whose OAM states are respectively $|+k_1\rangle$ and $|+k_2\rangle$, coming out on side A is
\begin{equation}
    \left|\frac{i}{2}\left(a_{k_{1}} b_{k_{2}} e^{i\left(\phi_{k_{1}}+\psi_{k_{2}}\right)}+a_{k_{2}} b_{k_{1}} e^{i\left(\psi_{k_{1}}+\phi_{k_{2}}\right)}\right)\right|^{2} > 0.
\end{equation}
In such cases, we discard the photon pair and regenerate a new one.

\subsubsection{Computational cost}
As confirmed in Eq. \eqref{eq:phom_rate}, the output OAM states of the Hong-Ou-Mandel effect depend on the input parameters $a_k, b_k, \phi_k$, and $\psi_k$.
Since the players control these parameters, the computational cost depends on how this control is done.
How exactly they should be controlled is discussed in Sec. \ref{subsubsec:phom_jm}.

\subsubsection{Confidentiality}
In this study, we assume a situation where each player adjusts the SLM using only their own preference; that is, player A cannot take $B_i$ into account to determine $a_k$ and $\phi_k$, and vice versa.
Under this assumption, neither player is required to disclose their probabilistic preference to the other or to a third party.
Even though player A does not know player B's preference, he/she must make an assumption to determine the input parameter $a_i$. For now, we let player A assume that player B has the same preference and amplitude terms. This is a reasonable assumption when preferences are similar, but does not hold in general.
Specifically, using the amplitude term $a_i$, he/she can compute the joint selection probability matrix from which the output OAM states are sampled from, so he/she can optimize $a_i$ numerically by viewing the loss of the joint selection probability matrix as a function of $a_i$.
The optimization of the amplitude terms are conducted numerically using the SLSQP optimizer.

Moreover, the Hong-Ou-Mandel effect allows them to avoid conflicts without having to inform the other party of which option they have selected. Therefore, with Pure HOM, the players' preferences and their choices are highly secure.
Also, there is no need to trust a third party since the whole procedures can be carried out between the two players. As we will see later, this property is unique to Pure HOM and cannot be achieved by the other quantum sampling method we propose in this paper.

\subsubsection{Joint selection probability matrix}\label{subsubsec:phom_jm}
By using Eq. \eqref{eq:phom_rate}, we can calculate the probability of player A selecting option $i$ and player B selecting option $j$ as
\begin{gather}
  p_{i,j} = \frac{r_{i,j}}{\sum r_{i,j}},\label{eq:r_to_p}\\
  r_{i,j} =\frac{1}{4}\left(a_{i}^{2} b_{j}^{2}+a_{j}^{2} b_{i}^{2}-2 a_{i} a_{j} b_{i} b_{j} \cos \left(\theta_{i}-\theta_{j}\right)\right).\label{eq:rij}
\end{gather}
This joint probability depends on the input parameters, and this section analyzes the characteristics of the joint selection probability matrix that consists of $p_{i,j}$ and discusses what parameters the players should use.
Note that we assume that each player controls the SLM using only his/her own preference.

We can confirm from Eqs. \eqref{eq:r_to_p}, \eqref{eq:rij} that the joint selection probability matrix is symmetric for any input parameters $a_k, b_k, \phi_k$, and $\psi_k$.
This means that the satisfied preferences of players A and B on option $i$, that is, $\pi_A(i)$ and $\pi_B(i)$, are always equal.
Thus, when players A and B have similar preferences, Pure HOM is likely to result in a low loss, but it is expected to work very poorly when they have reversed preferences.

Regarding situations where the players have the same preferences, the following theorem holds.
\begin{theorem}\label{thm:hom_3players}
  When there are three options, if the players have the same preferences and all the popularities $S_i$ are less than 1, by setting the amplitude terms $a_k, b_k$ as follows:
  \begin{align}\label{eq:phom_amps}
    \begin{split}
      a_{1}^{2}: a_{2}^{2}: a_{3}^{2}& = b_{1}^{2}: b_{2}^{2}: b_{3}^{2}=\\
      \frac{\sin ^{2} \frac{\theta_{2}-\theta_{3}}{2}}{1-2 A_{1}}: &\frac{\sin ^{2} \frac{\theta_{3}-\theta_{1}}{2}}{1-2 A_{2}}: \frac{\sin ^{2} \frac{\theta_{1}-\theta_{2}}{2}}{1-2 A_{3}},
    \end{split}
  \end{align}
  the resulting joint selection probability matrix achieves the theoretical minimum loss.
\end{theorem}
\begin{proof}
  \normalfont
  When $a_i^2=b_i^2$, Eq. \eqref{eq:rij} can be rewritten as
  \begin{equation}
    r_{i,j} = a_i^2a_j^2\sin^2{\frac{\theta_i-\theta_j}{2}}.
  \end{equation}
  Now, let $\hat{\pi}_A(i)$ be the ``unnormalized satisfied preference,'' which is defined by
  \begin{equation}
    \hat{\pi}_A(i) = \sum_{j=1}^N r_{i,j}.
  \end{equation}
  Also, the unnormalized satisfied preference for player B is defined by
  \begin{equation}
    \hat{\pi}_B(j) = \sum_{i=1}^N r_{i,j}.
  \end{equation}
  With the relation Eq. \eqref{eq:r_to_p}, it follows that
  \begin{equation}
    \pi_A(i) = \frac{\hat{\pi}_A(i)}{\sum_{i=1}^N \hat{\pi}_A(i)}, \quad \pi_B(j) = \frac{\hat{\pi}_B(j)}{\sum_{i=1}^N \hat{\pi}_B(j)}.
  \end{equation}
  Using the amplitude terms described in Eq. \eqref{eq:phom_amps}, we get
  \begin{gather}
    \hat{\pi}_A(1) = a_1^2a_2^2\sin^2{\frac{\theta_1-\theta_2}{2}}+a_1^2a_3^2\sin^2{\frac{\theta_1-\theta_3}{2}}\\
    = T(1-2A_3)a_1^2a_2^2a_3^2 + T(1-2A_2)a_1^2a_2^2a_3^2\\
    = 2Ta_1^2a_2^2a_3^2(1-A_2-A_3)\\
    = 2Ta_1^2a_2^2a_3^2\times A_1.
  \end{gather}
  Here,
  \begin{equation}
    T = \frac{\sin ^{2} \frac{\theta_{2}-\theta_{3}}{2}}{1-2 A_{1}} + \frac{\sin ^{2} \frac{\theta_{3}-\theta_{1}}{2}}{1-2 A_{2}} + \frac{\sin ^{2} \frac{\theta_{1}-\theta_{2}}{2}}{1-2 A_{3}}.
  \end{equation}
  Similarly, it follows that
  \begin{gather}
    \hat{\pi}_A(2) = 2Ta_1^2a_2^2a_3^2\times A_2, \quad \hat{\pi}_A(3) = 2Ta_1^2a_2^2a_3^2\times A_3.
  \end{gather}
  Normalizing $\hat{\pi}_A(i)$ leads to
  \begin{equation}
    \pi_A(i) = A_i.
  \end{equation}
  Therefore, the loss for player A is zero, and the same argument can apply to player B, which results in
  \begin{equation}
    L = 0.
  \end{equation}
  \qed
\end{proof}
This property cannot be realized by OAM Attenuation, that will be explained in the subsequent Sec. \ref{subsec:oam_attenuation} and other simple sampling methods shown in Sec. \ref{sec:demo}, highlighting the importance of Pure HOM.
Furthermore, in Sec. \ref{sec:demo}, we present the results of numerical simulations that show the near-optimality of Pure HOM under less restricted conditions, including when the number of options is more than three and when one of the popularities $S_i$ is greater than 1.
There, Pure HOM is found to be quite effective as long as the players have the same preferences.

\subsection{Orbital Angular Momentum Attenuation}\label{subsec:oam_attenuation}
\subsubsection{Implementation}\label{subsubsec:sr_imp}
This method, which we call ``Orbital Angular Momentum Attenuation (OAM Attenuation),'' also utilizes quantum interference of orbital angular momentum, and the quantum system to be considered is proposed by Amakasu {\it{et al.}} \cite{amakasu2021conflict}, but for a different purpose.
As with the case in Pure HOM, we perform probabilistic decision-making by mapping the observed OAM number to the index of choice.
We use a system described in Fig. \ref{fig:two_systems}, and the whole system works in the following way.

\begin{figure*}[htp]
  \includegraphics[width=13.0cm]{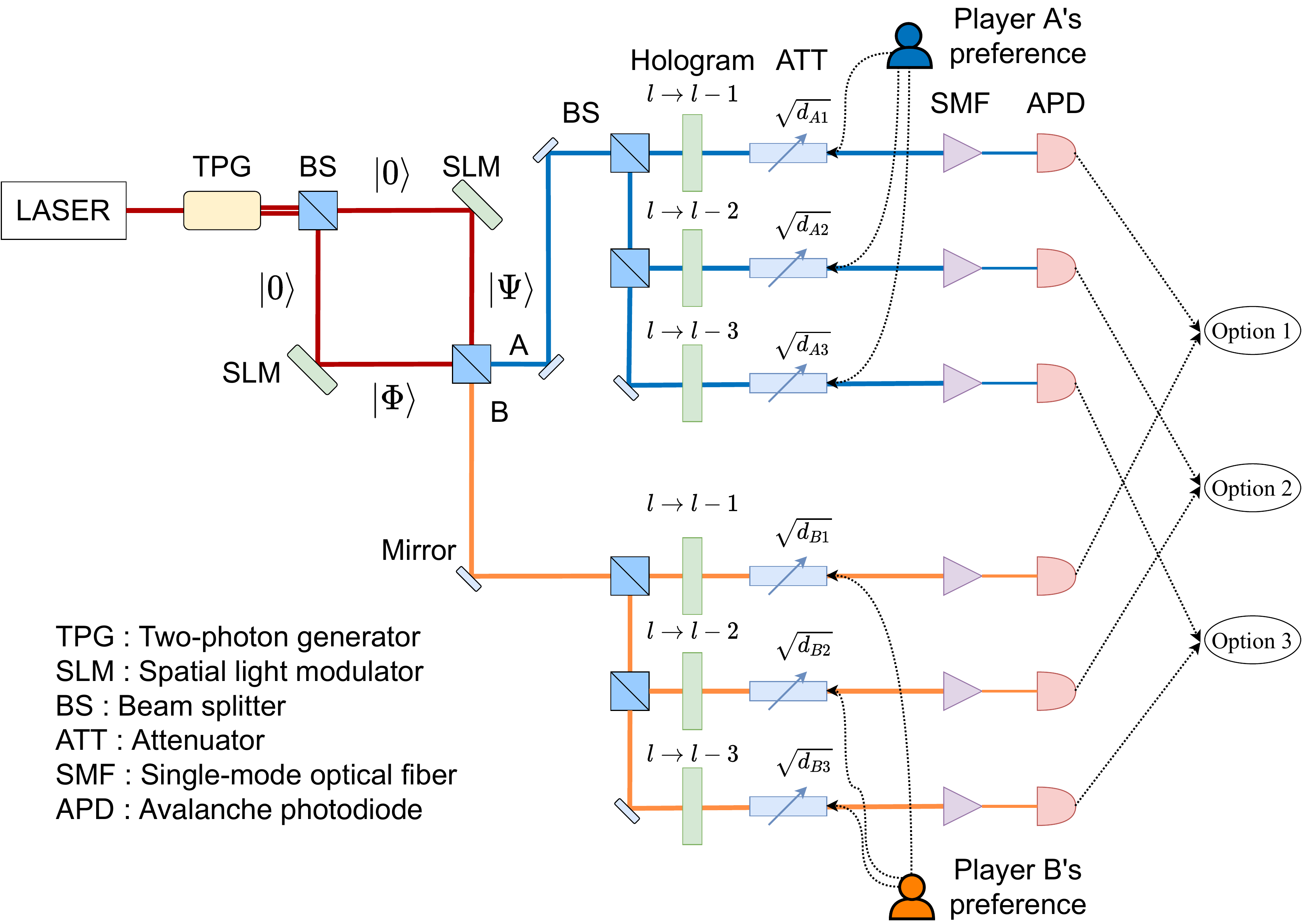}
  \caption{
  Implementation of OAM Attenuation. 
  Each player embeds their probabilistic preference to the attenuators. 
  The detected OAM number is mapped to the index of the option.
  }
  \label{fig:two_systems}
\end{figure*}

After a photon pair is generated by a two-photon generator, it is split into two paths by a beam splitter, and OAM states of the photons are adjusted by SLMs as follows:
\begin{equation}\label{eq:source_input_states}
    |\Phi\rangle=\frac{1}{\sqrt{K}} \sum_{k=1}^{K} e^{i \phi_{k}}|+k\rangle, \quad|\Psi\rangle=\frac{1}{\sqrt{K}} \sum_{k=1}^{K} e^{i \psi_{k}}|-k\rangle.
\end{equation}
In \cite{amakasu2021conflict}, only phase modulations are applied, and we follow this setup in OAM Attenuation. This corresponds with the situation where we set
\begin{equation}
    a_k = b_k = \frac{1}{\sqrt{K}}
\end{equation}
in Eqs. \eqref{eq:phom_phi} and \eqref{eq:phom_psi}.
As a result, the probability of detecting $|+i\rangle$ on side A and $|-j\rangle$ on side B is
\begin{gather}
  r_{i,j} = \frac{1}{K^{2}} \sin ^{2}\left(\theta_{i}-\theta_{j}\right),\\
  \theta_k = \frac{\phi_k-\psi_k}{2}.
\end{gather}
Here again, we can confirm that the observation probability is always zero when $i=j$. Also, there are cases where two photons come to the same side. Then, we discard the photon pair and regenerate a new one.

After the Hong-Ou-Mandel effect, each photon is sent to an attenuation system owned by each player, as shown in Fig. \ref{fig:two_systems}. There, the photon is divided into $N$ paths by beam splitters, and in each path, a phase factor of $e^{il_{HG}\theta}$ is added to the state $|+l\rangle$, which changes the state to $|+l+l_{HG}\rangle$, by a hologram. 
Then, the probability amplitude is reduced by an attenuator, and only an $l=0$ photon is filtered through by a single-mode optical fiber. If the photon is detected by a photodetector placed in the same line as the hologram with the phase factor $e^{il_{HG}\theta}$, the OAM of the incoming photon is revealed to be $l=-l_{HG}$.

If there are $N$ options, holograms whose phase factors are respectively $-1, -2, \ldots, -N$, are used. For example, when the number of options $N$ is three, three holograms that transform $|l\rangle$ to $|l-1\rangle$, $|l-2\rangle$ and $|l-3\rangle$, respectively, are placed.

Now, in the $i$-th path, an attenuator with the attenuation rate $\sqrt{d_i} \quad (0 \leq d_i \leq 1)$ is placed after the hologram and before the single-mode optical fiber.
In the end, if the input photon of the attenuation system carries the same probability amplitude for each OAM, the detection rate of each OAM is denoted by
\begin{equation}
  \frac{d_i}{N^2}.
\end{equation}
The sum over all $i$ is not equal to unity because some photons are lost by the attenuators and the fibers.

The detected OAM number is mapped to the index of the option. As a result, the selection probability of the $i$-th machine is $d_i$ as a result of the attenuation system.

In this study, the part of the system that generates photon pairs using the Hong-Ou-Mandel effect is considered as the ``source,'' and we assume that this is controlled by a neutral third party with no knowledge of the players' preferences. If only one of the players had the authority to manipulate the source, they would be able to adjust $|\Phi\rangle$ and $|\Psi\rangle$ so that some selection pair could happen more at the source than other pairs.
The generated pair of photons are then injected into the system shown in Fig. \ref{fig:two_systems}, each controlled by player A and B, respectively.

Each player embeds information about his/her own probabilistic preference in the attenuators in his/her system. Specifically,
player A adjusts the attenuation rates of the attenuators in his/her system to
\begin{equation}
  \sqrt{d_{Ai}} = \sqrt{A_i}.
\end{equation}
Similarly, player B sets
\begin{equation}
  \sqrt{d_{Bi}} = \sqrt{B_i}.
\end{equation}
Each player chooses the option whose index is equal to the OAM number they have observed.

\subsubsection{Computational cost}\label{subsubsec:sr_cost}
One remarkable aspect of this system is that it requires zero computational cost since the player only needs to set the attenuation rates. As we will see later, the computational cost of independently computing the effective joint selection probability matrix that this sampling method naturally converges to would be $O(N^2)$.
However, we need no computational cost if we need only to sample from it with the quantum system.
Note that, even though the players do not need to calculate anything, they sometimes have to repeat the observation process because the photon pairs are lost probabilistically due to two reasons. 
The first reason is the loss at the source level. We ignore cases where two photons come out on the same side of the beam splitter at the source. The second reason is the absorption by the attenuators and the filtering through the optical fiber. 
The players can detect only photons that pass through the attenuators and the single-mode optical fiber.

\subsubsection{Confidentiality}
Since each player sets only his/her own attenuators, there is no need to disclose their preferences.
In addition, thanks to the HOM effect, selection conflicts do not occur in principle, and thus conflict avoidance is achieved without each party having to communicate their own selection to the other.
However, they need to have a way to make sure that both of them detected a photon because each of them alone cannot discriminate the following two situations.
\begin{itemize}
    \item Both players detected a photon, which means that the joint selection is valid.
    \item One of the players detected a photon, but the other photon is absorbed by the attenuator, which means that the joint selection is invalid.
\end{itemize}
If they have a direct connection, they can just ignore the cases where one of them does not detect a photon. 
If they do not have a direct connection, a third person or system that executes the joint decision only when both players send their choices is needed.
Another pitfall is that the final value of the joint selection probability matrix to which this sampling method converges to depends not only on the players' preferences, but also on the phase settings of the source part, as we will examine in the next section.
Thus, the player has to trust that the third party managing the source determines the phase from a fair distribution.

\subsubsection{Joint selection probability matrix}
When there are $N$ options and the input states of the source are set as in Eq. \eqref{eq:source_input_states}, the general formula of the elements $p_{i,j}$ in the joint selection probability matrix is denoted by
\begin{gather}
  p_{i,j} = \frac{r_{i,j}}{\sum r_{i,j}},\\
  r_{i,j} = \frac{1}{K^{2}} \sin ^{2}\left(\theta_{i}-\theta_{j}\right) \times A_i \times B_j,\\
  \theta_k = \frac{\phi_k-\psi_k}{2}.
\end{gather}
Although the values of the element in the joint selection probability matrix depend on the source's phase settings $\theta_i$ as described above, this sampling method converges to the following joint selection probabilities over many repetitions, assuming that the source determines the value of $\sin^2(\theta_i-\theta_j)$ uniformly at random:
\begin{gather}
  p_{i,j} = \frac{\tilde{r}_{i,j}}{\sum \tilde{r}_{i,j}},\label{eq:sr_p}\\
  \tilde{r}_{i,j} = A_i \times B_j.\label{eq:sr_r}
\end{gather}
This is actually the same joint probability matrix that is introduced as ``Simultaneous Renormalization'' in \cite{shinkawa2022optimal}.
If this joint selection probability matrix were to be computed on a computer, it would require a computational cost of $O(N^2)$. 
However, if the sampling is carried out by the quantum system described in Fig. \ref{fig:two_systems}, as explained in Sec. \ref{subsubsec:sr_cost}, the cost is 0.
Although the quantum system guarantees low computational cost and high confidentiality, OAM Attenuation generally cannot achieve the optimal loss, as will be discussed in the following section.

\section{Performance comparison}\label{sec:demo}
\subsection{Objectives}
In this section, we compare the losses under various preference settings to clarify the extent to which the joint selection probability matrices of the proposed methods approximate the optimal joint selection probability matrix.
A total of five models were compared. 
The first two models were introduced in Sec. \ref{sec:heuristics}, that is, Pure HOM and OAM Attenuation.
The other models are Random Order, Uniform Random and the optimal joint selection probability matrix. 
Details of Random Order are explained in the subsequent Sec. \ref{subsec:random_order}. 
Uniform Random is a method that samples cases with no selection conflicts with equal probabilities.
As Eqs. \eqref{eq:sr_p}, \eqref{eq:sr_r} imply, the joint selection probability matrix that OAM Attenuation converges to is the same as the one introduced as ``Simultaneous Renormalization'' in Ref. \cite{shinkawa2022optimal}, albeit without considering the properties of a physical implementation, such as the absorption of photons.
Since we have already compared Uniform Random, Random Order, Simultaneous Renormalization, and the optimal joint selection probability matrix in the previous study, this section focuses primarily on the performance of Pure HOM. To be clear about the contribution of this paper to OAM Attenuation, it is the analysis of the implementation method, computational cost, confidentiality, and properties of the joint selection probability matrix it converges to, as presented in Sec. \ref{subsec:oam_attenuation}, and this section uses it to compare with Pure HOM. 

\subsection{Random Order}\label{subsec:random_order}
One straight-forward and classical method to realize conflict-free decision-making is what we call ``Random Order.'' 
It is similar to the random priority mechanism proposed by Abdulkadiro{\u{g}}lu {\it{et al.}}  \cite{abdulkadirouglu1998random}, except that it takes into account probabilistic preferences. 
A big advantage of it is its simplicity.
First, players decide uniformly at random in which order they choose options. 
Then, based on the order, the first player makes a probabilistic decision according to his/her preference. 
For there to be no choice conflict, the first player notifies the second player which option has already been chosen.
The second player then configures the preference of the already-selected option to zero, and normalizes his/her preference so that the sum of the remaining probabilities becomes 1.
After that, the second player executes a probabilistic choice based on his/her preference.
Note that this method can sometimes fail when there is a zero preference. For example, in an extreme case, when the number of options $N$ is two, and the players' preferences are $A_1=A_2=0.5, B_1=1$, and $B_2=0$. 
If they decide player A to be the first and he/she chooses option 1, player B will have no options with a positive preference to select, and the algorithm stops.

This simple algorithm can also reduce the problems with the optimal matrix.
Regarding the first problem, that is, the computational cost, the first player does not need to make any calculations. 
However, the second player needs to normalize his/her preference after setting the preference of the already selected option to 0, which requires the computational cost of $O(N)$.

As for the second problem, neither player is required to directly disclose their probabilistic preference to the other, but the first player has to tell the second player which option he or she has chosen to avoid decision conflict. This will indirectly expose their preference profiles over many trials.
In addition, it is an undesirable property if the two players do not trust each other or have limited means of communication.

Finally, iterating this sampling method leads to convergence to a certain joint selection probability matrix, and the general formula of its elements $p_{i,j}$ can be expressed as follows:
\begin{equation}
  p_{i,j} = \frac{1}{2}\left(A_i \times \frac{B_j}{1-B_i} + B_j \times \frac{A_i}{1-A_j}\right).
\end{equation}
The first term corresponds to the probability of player A selecting option $i$ and player B selecting option $j$ under the condition that player A draws first, and the second term corresponds to the same probability under the condition that player B draws first.
We can easily confirm that these joint probabilities give the optimal loss when the number of options $N$ is two.
However, in more general cases where $N \geq 3$, the joint selection probability matrix of Random Order cannot achieve the optimal loss, except in special circumstances, such as when all the preferences are equal.

\begin{figure*}[htbp] 
  \includegraphics[width=13.0cm]{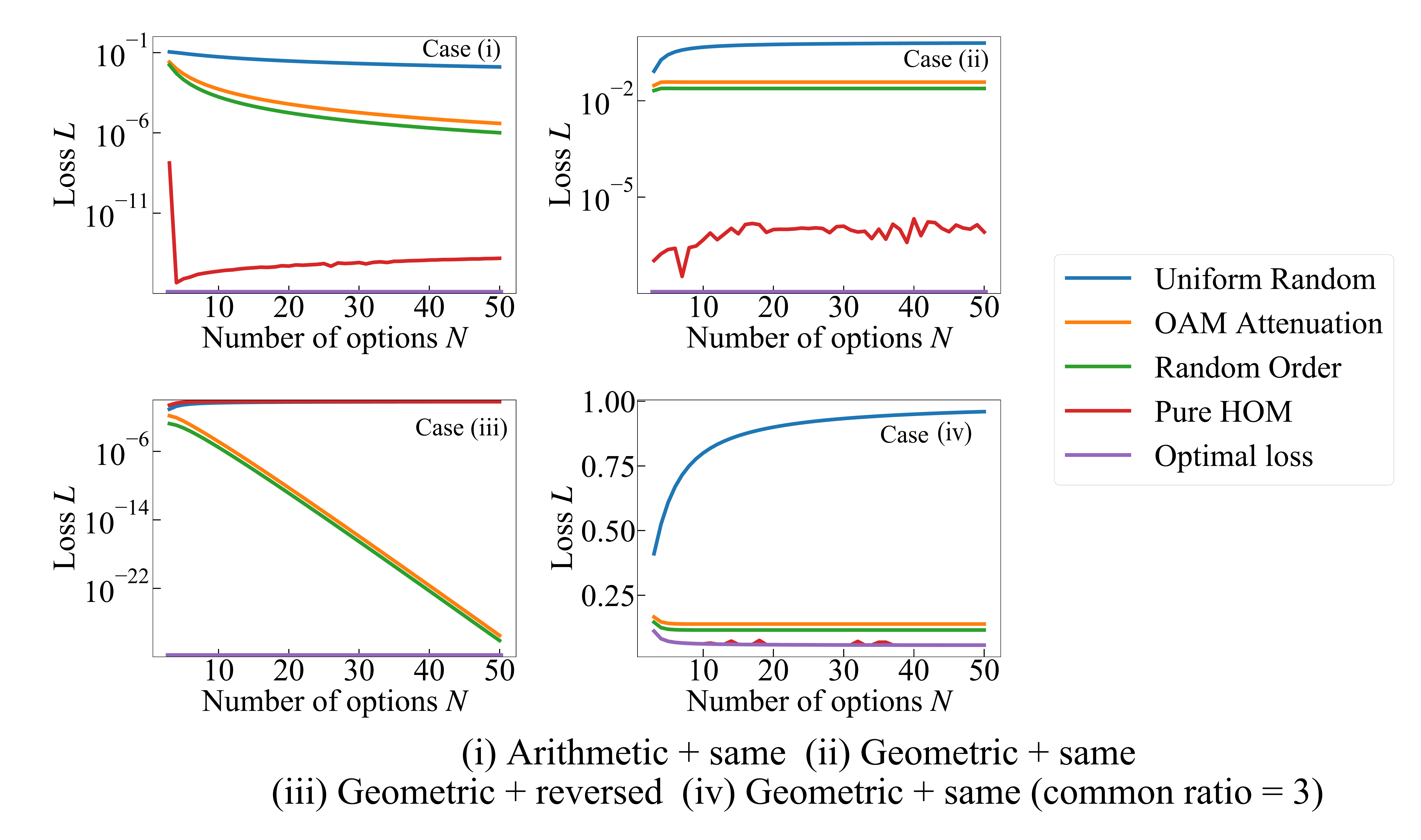}
  \caption{
  Loss comparison. $Y$ axes are log scale for cases (i), (ii) and (iii). 
  Lines for the optimal satisfaction matrix overlap with $X$ axes for these cases to show the loss is zero.
  (i)~arithmetic progression + same preference, (ii)~geometric progression with common ratio 2 + same preference, (iii)~geometric progression with common ratio 2 + reversed preference,
  (iv)~geometric progression with common ratio 3 + same preference. For cases (i)--(iii), the minimum loss is zero, while for case (iv), it is greater than zero.}
  \label{fig:main_result}
\end{figure*}

\subsection{Comparison of the heuristics}\label{subsec:comp}
The preference settings used in this section are the same as the ones used in Ref. \cite{shinkawa2022optimal}. 
Considering 2-player $N$-option ($N=3, 4, \ldots, 50$) situations, we compare the loss $L$ of Uniform Random, OAM Attenuation, Random Order, Pure HOM, together with the optimal joint selection probability matrix under the following four preference settings. 
\begin{enumerate}[label=(\roman*)]
  \item Arithmetic progression and same preference.\\
  \begin{align}
    \begin{split}
    A_1: A_2:\cdots : A_N = B_1:B_2:\cdots :B_N \\
    = (1:2:\cdots : N) / c_1, \quad c_1 = \frac{(N+1)N}{2}.
    \end{split}
  \end{align}
  \item Modified geometric progression with common ratio 2 and same preference.\\
  \begin{align}
    \begin{split}
    A_1: A_2:\cdots : A_N = B_1:B_2:\cdots :B_N \\
    = (1:1:2:\cdots : 2^{N-2}) / c_2, \quad c_2 = 2^{N-1}.
    \end{split}
  \end{align}
  \item Modified geometric progression with common ratio 2 and reversed preference.\\
  \begin{align}
    \begin{split}
    A_1: A_2:\cdots : A_N = B_N:B_{N-1}:\cdots :B_1 \\
    = (1:1:2:\cdots : 2^{N-2}) / c_3, \quad c_3 = 2^{N-1}.
    \end{split}
  \end{align}
  \item Geometric progression with common ratio 3 and same preference.\\
  \begin{align}
    \begin{split}
    A_1: A_2:\cdots : A_N = B_1:B_2:\cdots :B_N \\
    = (1:3:\cdots : 3^{N-1}) / c_4, \quad c_4 = \frac{3^N-1}{2}.
    \end{split}
  \end{align}
\end{enumerate}
Note that in cases (i)--(iii), the optimal satisfaction matrix achieves $L=0$ since $\forall i;S_i \leq 1$, whereas in case (iv), it is not possible to achieve $L=0$ since $S_N>1$.

Figure \ref{fig:main_result} shows how the loss $L$ for each of the five models changes as the number of options $N$ increases.
First, as expected, Uniform Random (blue line in Fig.~\ref{fig:main_result}) performs poorly under all of the conditions, as it does not take preferences into account. 

Next, as mentioned in the previous study, OAM Attenuation performs slightly worse than Random Order in all of the cases (orange and green lines in Fig.~\ref{fig:main_result}). 
This results in a trade-off between preference satisfaction and confidentiality. While OAM Attenuation has larger losses, Random Order requires the first player to inform the other player of his/her choice. 

Next, it is remarkable that Pure HOM has a very small loss compared to other heuristics in cases when both players have the same preferences (red line in Fig.~\ref{fig:main_result} cases (i), (ii) and (iv)).
In particular, for $N=3$, it can be mathematically proven that the loss is strictly zero when the maximum popularity $S_{\text{max}}$ is less than 1, as Theorem \ref{thm:hom_3players} suggests, but numerical calculations show that there is some residual loss.
This is due to the numerical errors which happen in the optimization of the amplitudes.

Over all, both quantum based sampling methods show promising performance. These methods do not require the direct calculation of a joint preference matrix or the disclosure of player's preferences, yet they can achieve very small losses. 

Remarkably, although case (iv) breaks $\forall i; S_i \leq 1$, the loss by Pure HOM is very close to the optimal loss. 
This implies that Pure HOM can work well in much less restricted conditions than those assumed in Theorem \ref{thm:hom_3players}, as long as the players have the same preferences. In the subsequent section, we examine the optimality of Pure HOM in more general preference settings.

On the other hand,  in case (iii), where the players have reversed preferences, Pure HOM performs worse than Uniform Random.
The reason is twofold.
First, in Pure HOM, we let player A assume that player B has the same preference and amplitude terms, but this assumption is strongly broken in case (iii).
Second, even if player A could take player B's preference into account, the joint selection probability matrix would always be symmetric in Pure HOM, so trying to satisfy the preference of one player will always lead to deviation from the preference of the other player.

\subsection{Optimality of Pure HOM}\label{subsec:phom_usage}
In Theorem \ref{thm:hom_3players}, we proved that when the players have the same preferences over three options and the popularities $S_i$ are all less than 1, Pure HOM can make the loss zero.
Moreover, case (iii) of Sec. \ref{subsec:comp} implies the possibility of Pure HOM being close to the optimal in more general settings.
This section further examines the loss $L$ of Pure HOM under other less restricted settings.

We are also interested in the efficiency of the physical sampling process. Therefore, we define the usage rate
\begin{equation}\label{eq:usage_rate}
  U = \sum_{i,j} r_{i,j}
\end{equation}
to measure the rate of photon pairs successfully used to make conflict-free decisions. In other words, $(1-U)$ of photon pairs are discarded in Pure HOM because the collective decision-making via Pure HOM works only when one photon is observed on side A and the other on side B in Fig. \ref{fig:mod_hom}.
Thus, bigger $U$ means that we are utilizing photon pairs efficiently.

Going beyond the cases studied in Sec.~\ref{subsec:comp}, we consider more general cases where the players have the same preferences and all the popularities $S_i$ are less than 1, but the number of options $N$ is 3--50.
For each number of options $N$, 1000 preferences are randomly chosen so that the maximum popularity is less than 1. 
Then, for each preference setting, the loss $L$ and the usage rate $U$ are calculated as a result of Pure HOM.
Finally, the average loss $L$ and the average usage rate $U$ are calculated over 1000 results. 
Figure \ref{fig:phom_performance_0loss}\subref{subfig:phom_loss_0loss} shows how the average loss $L$ changes as the number of options increases, and Fig. \ref{fig:phom_performance_0loss}\subref{subfig:phom_usage_0loss} shows the change in the average usage rate $U$.
The average losses for OAM Attenuation and Random Order are also presented in Fig. \ref{fig:phom_performance_0loss}\subref{subfig:phom_loss_0loss} as references.

\begin{figure}[htbp] 
    \centering
    \subfloat[\label{subfig:phom_loss_0loss}Average loss $L$]{%
        \centering
        \includegraphics[width=7cm]{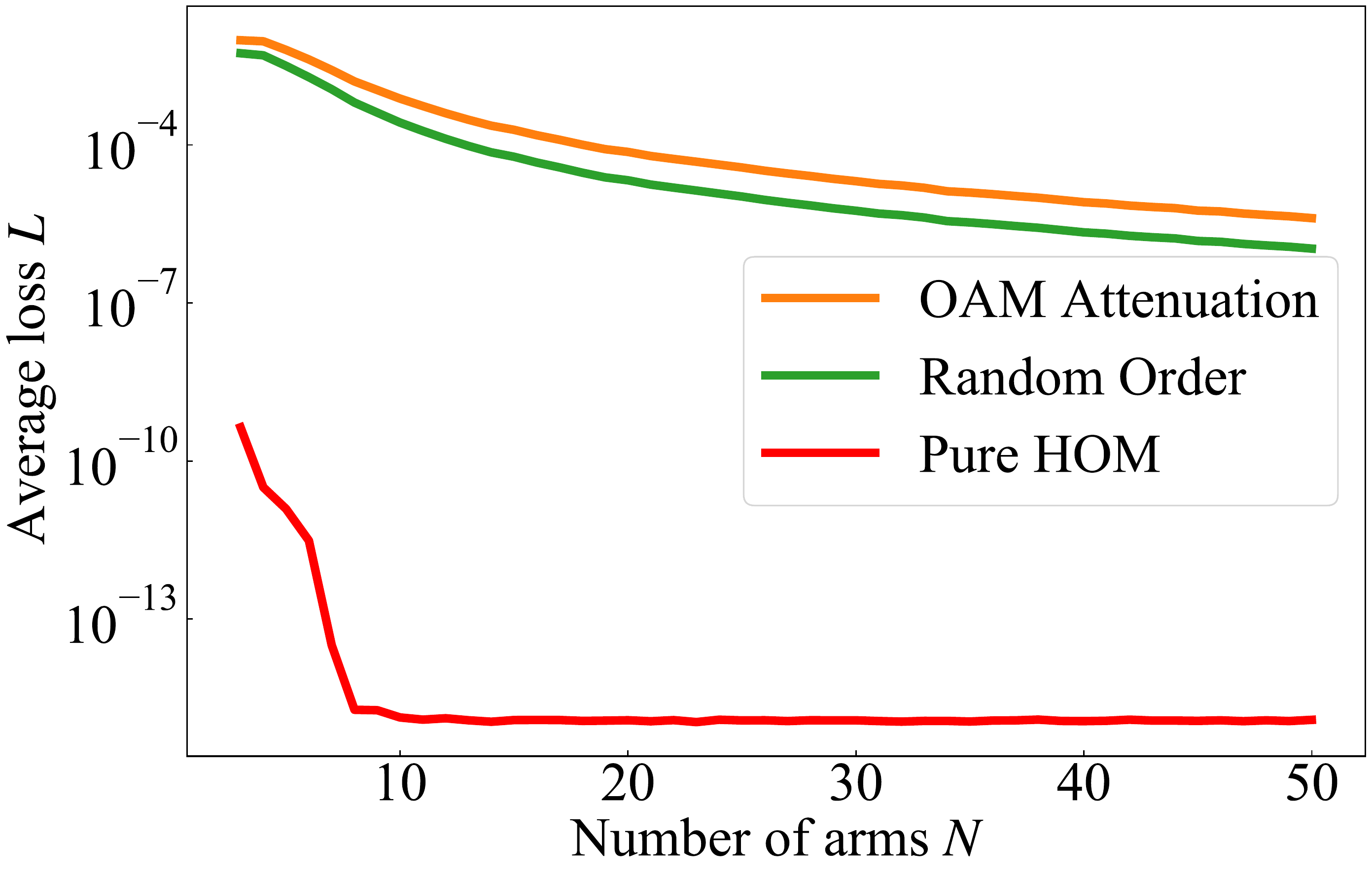}
    }\
    \subfloat[\label{subfig:phom_usage_0loss}Average usage rate $U$]{%
        \centering
        \includegraphics[width=7cm]{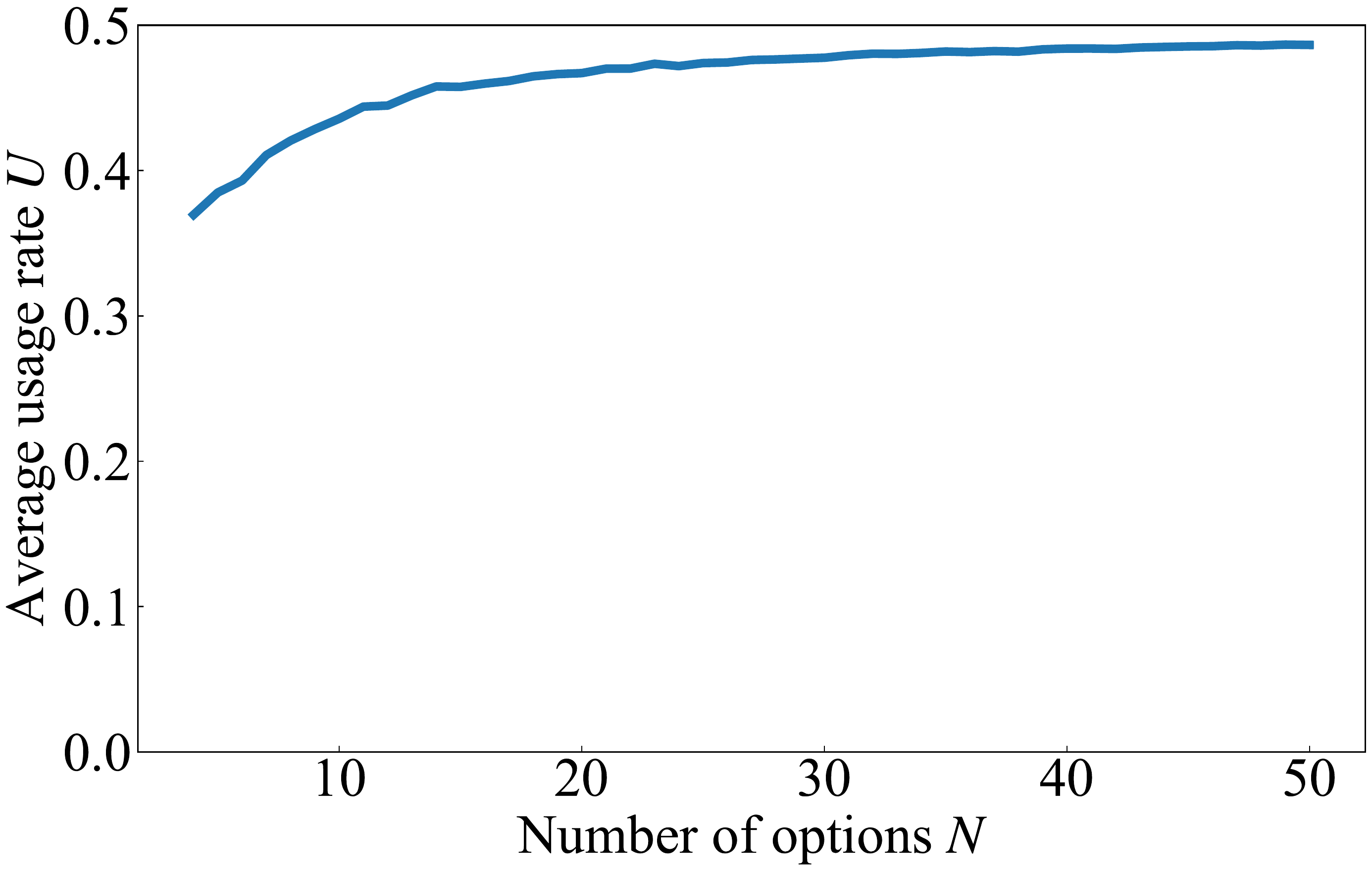}
    }
  \caption{ 
  The performance of the three sampling methods for general cases.
  (a) The average loss and (b) photon usage rate of Pure HOM calculated from $1000$ randomly generated symmetric preference profiles with maximum popularity less than 1.
  From (a), the loss by Pure HOM is quite close to zero, which is the theoretical minimum loss. Moreover, from (b), the usage rate is near 0.5, meaning that photon pairs are used efficiently.
  }
  \label{fig:phom_performance_0loss}
\end{figure}
\begin{figure}[htbp] 
    \centering
    \subfloat[\label{subfig:phom_maape_non0loss}]{%
        \centering
        \includegraphics[width=7cm]{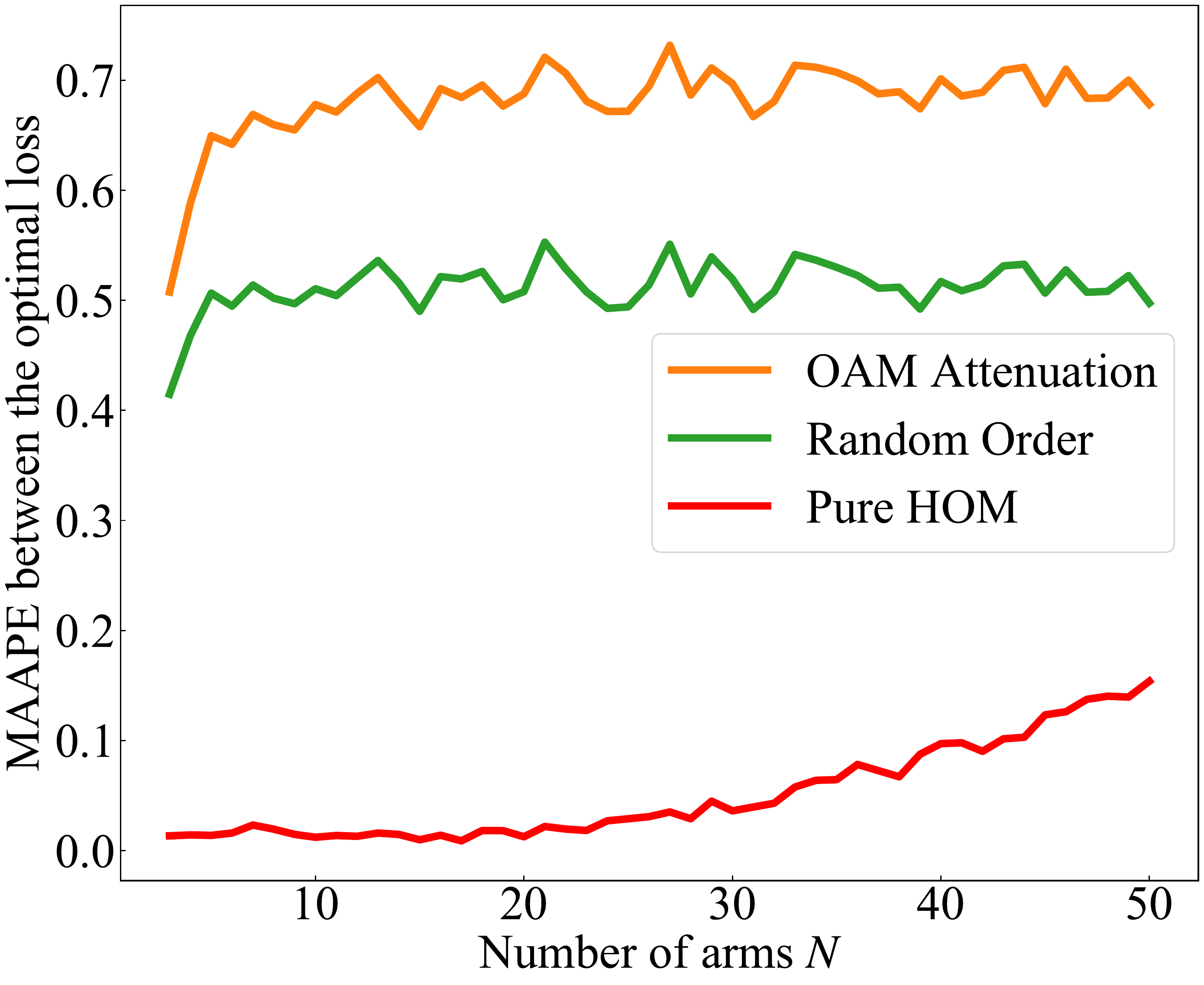}
    }\
    \subfloat[\label{subfig:phom_usage_non0loss}]{%
        \centering
        \includegraphics[width=7cm]{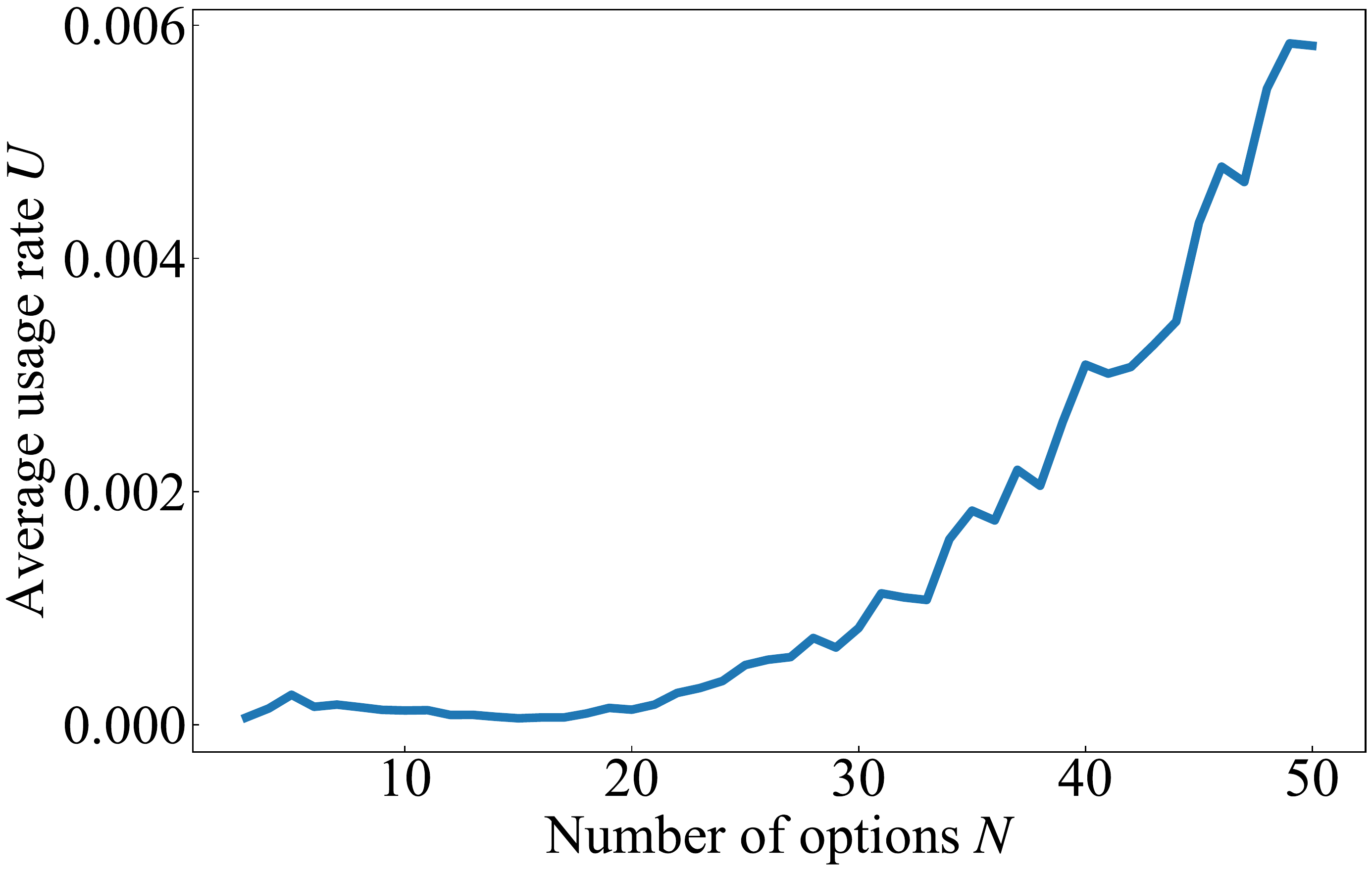}
    }
  \caption{
  The performance of the three sampling methods for general cases.
  (a) MAAPE between the optimal loss and (b) photon usage rate of Pure HOM calculated from $1000$ randomly generated symmetric preference profiles with maximum popularity greater than 1. 
  From (a), it is remarkable that MAAPE of Pure HOM is much closer to zero especially with small $N$. However, from (b), the usage rate is quite small, meaning that many photon pairs are needed to accomplish a collective decision.
  }
  \label{fig:phom_performance_non0loss}
\end{figure}

The average loss stays small for all numbers of options.
The losses for a smaller number of options become relatively bigger because
\begin{enumerate}
  \item The scale of each preference is bigger compared to the cases where $N$ is large, and so is the loss.
  \item There is more chance of the biggest preference being close to 0.5, which destabilizes the numerical optimization.
\end{enumerate}
Also, if we look at how the usage rate $U$ changes, it remains above about 0.35 and for bigger number of options, it approaches 0.5, meaning that about half of the photon pairs generated can be utilized.

Finally, we examine cases where the players have the same preferences, but the maximum popularity is greater than 1, meaning even the optimal satisfaction matrix cannot achieve 0-loss.
However, we must be careful about what metric to use in this case.
The average loss is not an appropriate metric to measure the near-optimality when the maximum popularity $S_{max}$ is greater than 1.
This is because the value of the optimal loss varies depending on $S_{max}$ (refer to Eq. \eqref{eq:minloss_greater_than_1} for the detailed formula).
Even if the average loss of Pure HOM over 1000 profiles is close to the average of the optimal losses, the value at each point can significantly deviate.
For example, suppose there is a preference profile with the optimal loss of $10^{-7}$, and the loss of Pure HOM for the same preference profile is $3\times 10^{-7}$.
Although there is a three-times gap between them, if another preference profile with the optimal loss of 0.1 exists, the error in the first profile is almost completely ignored when averaging the two.
Indeed, when $N=50$ and $S_{max}=1.0$, the optimal loss is 0, but when $S_{max}=2.0$, it is about 0.51, indicating that the scale problem is a critical one.

Accordingly, we use a new metric that reflects the relative error between the Pure HOM loss and the optimal loss rather than the average of the losses when the maximum popularity $S_{max}$ is greater than 1.
The most intuitive one is the mean absolute percentage error (MAPE) shown below, but MAPE cannot be used in this case because it becomes infinitely large when the optimal loss is close to 0.
\begin{equation}
    MAPE = \frac{1}{T}\sum_{i=1}^T \left|\frac{\tilde{L}_i - L_i}{\tilde{L}_i}\right|.
\end{equation}
Here, $T$ is the number of samples, and $\tilde{L}_i$ and $L_i$ respectively denote the optimal loss and the loss of a given sampling method (any one of Pure HOM, OAM Attenuation, or Random Order) for the $i$-th preference profile.
Instead, we quantify the closeness of the two losses over multiple points using the mean arctangent absolute percentage error (MAAPE), which is proposed by Kim {\it{et al}}. to overcome the weakness of MAPE \cite{kim2016new}.
This metric takes the arctangent of the relative error so that the value of the metric does not diverge when the ground truth is close to 0.
The value can take the range 0--$\pi/2$, with the smaller value indicating that the losses are closer.
\begin{equation}
    MAAPE = \frac{1}{T}\sum_{i=1}^T \text{arctan}\left|\frac{\tilde{L}_i - L_i}{\tilde{L}_i}\right|.
\end{equation}

Figure \ref{fig:phom_performance_non0loss}\subref{subfig:phom_maape_non0loss} shows MAAPE between the optimal loss and the three sampling methods: OAM Attenuation, Random Order and Pure HOM, and Fig. \ref{fig:phom_performance_non0loss}\subref{subfig:phom_usage_non0loss} is the change in the average usage rate $U$.
The reason why the results appear noisy is due to the insufficient number of samples.
Nevertheless, we do not believe it is necessary to increase the number of samples to smooth the lines, since our focus here is to compare the optimality of each sampling method, not to estimate individual MAAPE for every $N$ accurately.
We set the number of samples to 1000 in this experiment due to the limitation of computation time.
The graph shows that MAAPE of Pure HOM is much smaller than that of OAM Attenuation or Random Order especially with small $N$, although it gradually deviates from zero when the number of options $N$ is large. 
This is because optimization becomes difficult as the number of parameters becomes larger.
Together with the result shown in Fig.  \ref{fig:phom_performance_0loss}\subref{subfig:phom_loss_0loss}, for both cases where the maximum popularity is greater than or less than 1, Pure HOM can achieve losses that are much closer to the theoretical minimum than the other two sampling methods.

However, when the maximum popularity is greater than 1, the usage rates are significantly lower, on the order of $10^{-3}$, as demonstrated in Fig. \ref{fig:phom_performance_non0loss}\subref{subfig:phom_usage_non0loss}, 
meaning that we have to discard a lot of pairs of photons.

\section{Conclusion}\label{sec:conclusion}
In this paper, we deal with a situation in which multiple players have probabilistic preferences and consider the problem of satisfying their preferences.
The previous study explicitly computed the  joint selection probability matrix that maximized players' satisfaction. 
However, there were two concerns with the previous approach: high computational cost and low confidentiality.
This paper proposes two sampling methods that are implemented in quantum ways, each of which converges to a particular joint selection probability matrix accomplishing a relatively low loss.
We examined the implementation method, computational cost, confidentiality, and the joint selection probability they converge to. 
Specifically, OAM Attenuation allows sampling with zero computational cost and also guarantees a high degree of confidentiality, as players do not need to disclose their preferences or choices.
We also showed that Pure HOM can exclude the necessity to trust a third party while reducing losses to near-optimal values in situations where players have the same preferences.
The property of favoring similar preferences is useful in many real situations.
For example, in the competitive multi-armed bandit problem, since the machines have fixed reward probabilities over time, the players' preferences are expected to converge to similar values for each machine.

Looking ahead to the physical implementation, which we are currently working on, errors in preparing the input photon states will have different degrees of impact on the final joint selection probability matrix for the two quantum sampling methods.
In OAM Attenuation, it is the attenuator part that reflects individual preferences, and the source part plays the role of conflict avoidance.
Even if the phase difference or the amplitude settings of the source part deviate slightly, the effect on the final joint selection probability matrix is expected to be relatively small because the conflict avoidance is still ensured and the preference reflection by the attenuators remains.
On the other hand, in the Pure HOM, if $a_1$ increases by $\epsilon$ and $a_2$ decreases by $\epsilon$ instead, all the values in the first and second row as well as in the first and second column of the joint selection probability matrix will change.
Therefore, the states of the input photons must be carefully controlled.

Future studies include the mathematical or theoretical understanding of why Pure HOM can achieve a loss quite close to the theoretical minimum when players have the same preferences. 
Moreover, examining the possibilities of realizing an efficient sampling method that yields the optimal joint selection probability matrix is an interesting future topic.
In the meantime, we considered the average joint selection probability matrix for each sampling method assuming an infinite number of repetitions.
However, for example, the joint selection probability matrix for Random Order varies greatly depending on the order of players, especially in a small finite number of repetitions. In the case of OAM Attenuation, it also varies depending on the setting of $\theta_i$ at the source point.
The evaluation of such sample-wise variance is also a future topic.
\section*{Acknowledgements}
This work was supported in part by the CREST project (JPMJCR17N2) funded by the Japan Science and Technology Agency and Grants-in-Aid for Scientific Research (JP20H00233) funded by the Japan Society for the Promotion of Science.

\bibliography{reference}

\end{document}